\documentclass[journal,onecolumn]{IEEEtran}
\usepackage{caption}
\usepackage{authblk}
\usepackage{amsthm}
\usepackage{hyperref}
\usepackage{changepage} 
\usepackage[super]{nth}
\usepackage{tabularx}
\usepackage{microtype}
\usepackage{mathtools}
\usepackage[subnum]{cases}
\usepackage{empheq}
\usepackage{multirow}
\usepackage{amssymb}
\usepackage{fancyhdr,lipsum}
\usepackage{cite}

\ifCLASSINFOpdf
\else
\fi
\begin{document}
%
\title{Prosumer\index{Prosumer} Behavior: \\ 
Decision Making\index{Decision making} with Bounded Horizon}
%
%
%

\author[1]{Mohsen~Rajabpour
\thanks{This work is supported in part by the U.S. National Science Foundation (NSF) under the grants 1541069 and 1745829.}
}
\author[2]{Arnold~Glass}
\author[3]{Robert~Mulligan}
\author[1]{Narayan~B.~Mandayam}

\affil[1]{WINLAB, Department of Electrical and Computer Engineering, Rutgers University, NJ, USA.}
\affil[2]{Psychology Department, Rutgers University, NJ, USA.}
\affil[3]{The Kohl Group, Inc., NJ, USA.}

%
%

\markboth{}%
{Shell \MakeLowercase{\textit{et al.}}: Bare Demo of IEEEtran.cls for IEEE Journals}
%



\maketitle

\begin{abstract}
Most studies of prosumer\index{Prosumer} decision making\index{Decision making} in the smart grid\index{Smart grid} have focused on single, temporally discrete decisions within the framework of expected utility theory\index{Expected Utility Theory} (EUT\index{EUT}) and behavioral theories such as prospect theory\index{Prospect Theory}. In this work, we study prosumer\index{Prosumer} decision making\index{Decision making} in a more natural, ongoing market situation in which a prosumer\index{Prosumer} has to decide every day whether to sell any surplus energy units generated by the solar panels on her roof or hold (store) the energy units in anticipation of a future sale at a better price. Within this context, we propose a new behavioral model that extends EUT\index{EUT} to take into account the notion of a bounded temporal horizon\index{Temporal horizon} over which various decision parameters are considered. Specifically, we introduce the notion of a bounded time window\index{Time window} (the number of upcoming days over which a prosumer\index{Prosumer} evaluates the probability that each possible price will be the highest) that prosumers\index{Prosumer} implicitly impose on their decision making\index{Decision making} in arriving at ``hold" or ``sell" decisions. The new behavioral model assumes that humans make decisions that will affect their lives within a bounded time window\index{Time window} regardless of how far into the future their units may be sold. Modeling the utility of the prosumer\index{Prosumer} using parameters such as the offered price on a day, the number of energy units the prosumer\index{Prosumer} has available for sale on a day, and the probabilities of the forecast prices, we fit both traditional EUT\index{EUT} and the proposed behavioral model with bounded time windows\index{Time window} to data collected from 57 homeowners over 68 days in a simulated energy market\index{Energy market}. Each prosumer\index{Prosumer} generated surplus units of solar power and had the opportunity to sell those units to the local utility at the price set that day by the utility or hold the units for sale in the future. For most participants\index{Participant}, a bounded horizon in the range of 4--5 days provided a much better fit to their responses than was found for the traditional (unbounded) EUT\index{EUT} model, thus validating the need to model bounded horizons imposed in prosumer\index{Prosumer} decision making\index{Decision making}.
\end{abstract}


%
\IEEEpeerreviewmaketitle

\section{Introduction}
%
%
%
%
\IEEEPARstart{T}{he} smart grid\index{Smart grid} is a twenty-first-century evolution of the traditional electric power grid resulting from the convergence of several factors including  growing demand for electricity and a variety of technological developments, environmental imperatives, and security concerns. Unprecedented changes over the past few years in the smart grid\index{Smart grid} have enabled power utility companies to construct a two-way communications infrastructure to monitor and manage devices across the power system.  Smart meters that communicate usage information directly to the power utility and smart appliances that allow consumers to easily schedule their usage at times of low demand are two such elements. Some utility companies have already implemented incentive programs to influence the timing of consumers' usage to better balance demand.  These and other elements of the smart grid\index{Smart grid} infrastructure such as renewable energy generation and storage technologies (e.g. solar panels, wind towers, batteries.) have been developed and are being deployed across the United States and internationally.

Two other important characteristics of the smart grid\index{Smart grid} are more distributed power generation capability (versus the traditional centralized power station model) and a larger role for renewable energy resources like wind and solar (and an eventual diminishing reliance on traditional fossil fuel-fired power plants). These two characteristics merge in the proliferation of customer-owned rooftop solar panels (and less frequently, wind turbines). Together with a smart metering technology called ``net metering," which allows the flow of electricity in both directions across the meter (from power grid to home and from home to power grid), a new widely distributed power generation infrastructure is emerging in which an individual, so equipped, can be both a consumer and a producer of electric power. These individuals have been labeled ``prosumers."\index{Prosumer} They not only consume energy from the power utility but they can also produce energy to power their own homes and can often generate surplus electricity that they can feed back into the grid for an energy credit or sell to the utility at the market price.  When also equipped with a storage battery, and multiplied by thousands or tens of thousands, these prosumers\index{Prosumer} become major players in the energy market\index{Energy market} and a significant consideration in the design and management of the smart grid\index{Smart grid}.  

How prosumers\index{Prosumer} behave in this emerging energy management\index{Scenario} scenario--their perceptions, motivations and decisions--will influence the success of the smart grid\index{Smart grid} \cite{7426734,7781198,6855114,8031035}. For example, how prosumers\index{Prosumer} make decisions about whether and when to sell the surplus energy generated by their solar panels or wind turbines will inform the design of optimal energy management systems. Prosumers\index{Prosumer} can optimize their profits by deciding on the amount of energy to buy or to sell \cite{7426734,7781198,6855114}. Thus, important technical and economic challenges emerge as prosumers\index{Prosumer} influence demand-side management.

In response to this realization, a substantial literature has emerged with the objective of modeling the behaviors of electricity prosumers\index{Prosumer} in smart grid\index{Smart grid} settings.  One early example called out the need to consider the behavior of large numbers of small scale prosumers\index{Prosumer} in smart grid\index{Smart grid} planning efforts \cite{5638967}. The authors outlined a framework for modeling this behavior, but did not develop such a model. Several other researchers have taken a more quantitative approach, using game-theory to analyze demand-side management and to model prosumer\index{Prosumer} decision behavior in smart grid\index{Smart grid} scenarios\index{Scenario}. A limitation of these papers is that they consider the prosumers' behavior to be consistent with EUT\index{EUT} \cite{5628271,6552997,6279592,101007hajj}. EUT\index{EUT} provides a framework for optimizing the amount of money earned within the narrow context-free constraints of that theory. However, EUT\index{EUT} has often been found to not accurately account for how people respond in real-life situations.

To more accurately describe prosumer\index{Prosumer} behavior in realistic smart  grid  scenarios\index{Scenario},  a few recent studies have utilized PT for analyzing demand- side management in the smart grid\index{Smart grid}. PT is a descriptive model of prosumer\index{Prosumer} behavior that provides an empirical and mathematical framework that represents subjective utility functions instead of the objective utility functions of EUT\index{EUT} \cite{KahTev,TevKah}. These more recent investigations \cite{7426734,7781198,6855114,8031035}, \cite{7852237,6895275,7360934} have achieved some success in modeling prosumer\index{Prosumer} behavior by assuming that their decision making\index{Decision making} follows the percepts of PT. To our knowledge however, no studies have observed the decision behavior of actual prosumers\index{Prosumer} in the field, or that of subjects emulating prosumer\index{Prosumer} behavior in a laboratory simulation of a smart grid\index{Smart grid} scenario\index{Scenario}.

The approach taken here is complementary to the insights provided by PT and extends the previous works in two regards:  (1) from a methodological perspective, the current study measures the behaviors of human participants\index{Participant} (prosumers\index{Prosumer}) in a realistic smart grid\index{Smart grid} scenario\index{Scenario}, and (2) from a modeling perspective, it considers the effect of the subjective time frame within which a prosumer\index{Prosumer} operates when making decisions about whether and when to sell their surplus energy in an ongoing market. 

Analyses of decision making\index{Decision making} under uncertainty have typically investigated scenarios\index{Scenario} in which a single decision is required among two or more alternatives at a fixed point in time. In the smart grid\index{Smart grid} scenario\index{Scenario} considered here, however, prosumers\index{Prosumer} are engaged in an energy market\index{Energy market} in which they must make a series of nonindependent decisions over an extended period. For example, if they are generating and storing surplus electric power, and the price at which they can sell that power to the utility company can vary significantly over time due to several factors, they might have to decide on a weekly, daily, or even on an hourly basis, when to sell to maximize their profit. 

Within the framework of EUT\index{EUT}, if someone has something to sell, there is no parameter corresponding to a finite horizon within which a sale must be made. Furthermore, there is no concept of a time window\index{Time window} within that horizon over which a seller chooses to compute the utility of a sale. Consequently, the concept of a time window\index{Time window} has only rarely been mentioned or tested. One notable exception is when the seller is selling their own labor. The concept of a time horizon has been discussed by economists in the context of research on optimal job search strategies, looking in particular at the question of whether a job seeker should accept a job at an offered salary or wait for a job at a higher salary. One laboratory study \cite{Cox1,Cox2} imposed a finite horizon on subjects search decisions, but unlike our 10-week study, sessions were only an hour in length.  Also, unlike our realistic smart grid\index{Smart grid} simulation, the authors  used an abstract economic decision task rather than a realistic job search scenario\index{Scenario}.  

Another area of research in decision making\index{Decision making} that shares some characteristics with the prosumer\index{Prosumer} scenario\index{Scenario} is investing. Like prosumers\index{Prosumer}, individual investors have an opportunity to make a series of decisions over an extended time horizon and to selling none, some, or all  their resources. Looking at investment decisions under risk, one topic that has received considerable study has been the so-called disposition effect--the tendency of investors to sell too soon stocks that have increased in value relative to their purchase price and to hold on too long to those that have declined in value \cite{Shef}. A similar effect is seen in buying and selling real estate \cite{Genes}. 

Another investing phenomenon in which a decision horizon can come into play is what Benartzi and Thaler have dubbed myopic loss aversion\index{Loss aversion} (MLA) \cite{Bena}. These authors proposed MLA, a combination of loss aversion\index{Loss aversion} and another concept known as mental accounting \cite{Tha1,Tha2}, as an explanation for the ``equity premium puzzle"--the fact that in investment markets, equities are priced at a much higher premium (relative to bonds and other more stable investments) than would be predicted by standard economic models. To the extent that this myopia is temporal in nature--i.e., due to adopting a short-term horizon when longer-term would be more appropriate, this work could prove helpful in explaining prosumer\index{Prosumer} behavior in an ongoing energy market\index{Energy market}.

In the present study, we examine prosumer\index{Prosumer} behavior extending over 10 weeks, making it the first meaningful test of the role of time windows\index{Time window} in the behavior of sellers. The results reveal that an elaborated EUT\index{EUT} model that computes the probability of selling at a higher price within a fixed time window\index{Time window} predicts human decision making\index{Decision making} in the context of energy management better than conventional EUT\index{EUT}, providing crucial understanding of prosumer\index{Prosumer} decision making\index{Decision making} in the smart grid\index{Smart grid}. The elaborated EUT\index{EUT} model provides a framework and mathematical method to model how people make decisions within a bounded time window\index{Time window} that extends classic EUT\index{EUT}, which assumes an unbounded time window\index{Time window}. 

The rest of this chapter is organized as follows. Section \ref{1.2} describes the study methodology. Section \ref{1.3} presents the mathematical representation of EUT\index{EUT} and EUT with a time window\index{Time window}, and a method for determining the minimum price at which a prosumer\index{Prosumer} should sell. In section \ref{1.4} we present the data fitting\index{Data fitting} results and draw some conclusions.  Finally, section \ref{1.5} provides some discussion of the implications of these findings.

 

\section {Experimental Design: Simulation of an Energy Market\index{Energy market}}\label{1.2}
The decision behavior of 57 household decision-makers was studied in a simulation of a 10-week, collapsing horizon, smart energy grid scenario\index{Scenario}. Homeowners were recruited through an online research participant\index{Participant} recruiting service (\url{www.findparticipants.com}). Based upon demographic data gathered by the recruiting service and on participants' responses to a prestudy screening survey, 60 participants\index{Participant} were recruited, all of whom met the following criteria: 
\begin{itemize}
\itemsep0em 
\item	Homeowner;
\item Decision maker in the household regarding energy utility services;
\item PC, tablet, or smartphone user with access to the Internet; and
\item Committed to participating in the study by playing a short ``energy grid game\index{Grid game}" every day (except holidays) for 10 weeks.
\end{itemize}

Three subjects dropped out of the study in the first week. The remaining 57 all completed the 10-week study.

\paragraph{Recruitment and Compensation}

Subjects were compensated for study participation in a manner that incentivized them to attend carefully and perform well. All subjects earned a fixed amount (\$100) for completing the study provided they responded on at least 65 of the days. In addition to this fixed amount, subjects kept the ``profit" they earned by selling their surplus energy during the course of the game. Note that all subjects accumulated the same amount of surplus energy on the same schedule over the 70 days. However, the amount of profit earned could range from about \$40 to \$120 depending upon how well they played the game, i.e., how adept they were at selling when energy prices were highest. Finally, to further incentivize subjects to attend carefully to their task, the top three performers were awarded bonuses of \$25, \$50, and \$100.

\paragraph{The ``Grid Game\index{Grid game}" Scenario\index{Scenario}}
The framework for the study was a simulated smart grid\index{Smart grid} scenario\index{Scenario} in which participants\index{Participant} acted as prosumers\index{Prosumer}. They were asked to imagine that they had solar panels installed on their property, a battery system to store surplus energy generated by their solar panels, and a means of selling their surplus energy back to the power company for a profit. 

 In the scenario\index{Scenario}, the subjects' solar panels generated a small amount of surplus electricity (i.e., over and above household power requirements) on most days of the study. For purposes of this simulation, we avoided references to megawatt or kilowatt hours and simply referred to the units of electricity as units. In an effort to simulate variability introduced by factors like weather and household demand for electricity, the number of surplus units generated on each day was sampled from a distribution in which one unit of surplus power was generated on 50\% of the 70 days; 2 units on 35\% of the days and zero units on the remaining 15\% of days. 
 
Surplus energy was stored in the subjects' hypothetical storage batteries. A total of 89 units were generated over the course of the 10-week study. Each day, subjects were given the opportunity to sell some or all the surplus energy they had accumulated in their storage system. The subjects' task was to maximize their profits by selling their surplus electricity when prices were highest.

\paragraph{Energy Price Data}

The price paid by the fictional power company for prosumer-generated electricity varied from day to day in a manner that loosely mimicked the variation in wholesale electricity prices. The prices were designed to approximate the range and variability in the wholesale energy market\index{Energy market} in the United States. The data source for modeling prices was monthly reports of wholesale energy prices published by the U.S. Energy Industry Association 
(e.g., \url{www.eia.gov/electricity/monthly/update/wholesale_markets.cfm}) 
referencing data from March and April, 2016.

A distribution of daily price per unit was created, comprising 15 discrete prices ranging from \$0.10 to \$1.50 in 10-cent increments. The shape of the distribution was positively skewed, emulating real-world prices, with a mean of \$0.60 and a modal price of \$0.50. The probability of each price occurring on any given day of the study is shown in Table \ref{t1}. We refer to this set of 15 probabilities as the single-day probabilities ($p_i$) for the 15 possible prices. The daily price offered to participants\index{Participant} for their surplus energy was determined by sampling 70 values from this distribution. After the initial set of 70 daily prices was generated, an offset was applied to weekday (Mon--Fri) and weekend (Sat--Sun) prices such that the average weekday prices (\$0.70) were 40\% higher than average weekend prices (\$0.41). This offset did not change the characteristics of the overall distribution of prices. 

\begin{table}[b]
\centering 
\caption{Single-day probabilities of each of the 15 possible wholesale electricity prices.}
\small
 \begin{tabularx}{\columnwidth}{XXXXXXXXXXXXXXXX} 
 \hline
 \tiny{Price} & 0.1 & 0.2 & 0.3 & 0.4 & 0.5 & 0.6 & 0.7 & 0.8 & 0.9 & 1.0 & 1.1 & 1.2 & 1.3 & 1.4 & 1.5   \\ [0.5ex] 
 \hline
 $p_i$ & 0.03 & 0.06 & 0.09 & 0.12 & 0.14 & 0.11 & 0.09 & 0.08 & 0.07 & 0.06 & 0.05 & 0.04 & 0.03 & 0.02 & 0.01   \\ [0.5ex]
 \hline
\end{tabularx}
\label{t1} 
\end{table}

\paragraph{Instructions to Participants\index{Participant}}
At the beginning of the study, participants\index{Participant} were informed about the amount of energy they would generate and the distribution of prices they could expect.  They received the following instructions: 

\begin{adjustwidth}{1cm}{1cm}

\textit{Here is some important information to help you to maximize your profits from selling your units:}

\begin{enumerate}
\itemsep0em 
\item	For the 10-week period of our study, the price per unit paid by our imaginary power company can range from \$0.10 (10 cents/unit) to \$1.50/unit, in 10-cent increments. So, there are 15 possible prices--\$0.10, \$0.20, \$0.30, etc., up to \$1.50 per unit.
\item	The average price across the 10 weeks will be \$0.60/unit (60 cents per unit).
\item	On average, the price paid on weekdays (Mon--Fri) tends to be somewhat higher than on weekends because demand is greater during the work week.
\item	Over the 70 days, your system will generate a total of 90 units of surplus electric power. For purposes of our game, you can assume that your battery system can easily store this much power and there is no appreciable decay over time.
\end{enumerate}

Here is how the game will work. Every day before noon Eastern Time you will get a Grid Game\index{Grid game} e-mail providing you with the following information:

\begin{itemize}
\itemsep0em 
\item The number of units of surplus energy generated by your solar panels yesterday.  This will always be 0, 1, or 2 units.
\item The total number of units stored in your storage battery. You'll start the game with 5 units.
\item The amount of profit you have earned so far by selling your surplus energy.
\item The price the power company is paying for any units you sell today.
\item A table showing the probability that each of the 15 prices will occur at least once between today and the end of the study. For example, you will see:
\begin{itemize}
\itemsep0em 
\item[$\circ$]  \$1.50 per Unit -- 30\% chance
\item[$\circ$]  \$1.40 per Unit -- 45\% chance
\item[$\circ$]  \$1.30 per Unit -- 78\% chance
\item[$\circ$]  ... 
\item[$\circ$]  \$0.20 per Unit -- 59\% chance
\item[$\circ$]  \$0.10 per Unit -- 35\% chance

\end{itemize}
\end{itemize}
This table of probabilities is also provided to help you maximize your profits. The probabilities are based on actual variability in energy prices in the real energy grid. \textit{They will change gradually as we move through the study, so be sure to check them every day.}

Your task each day will be to respond to the e-mail by answering the following question:  Do you want to sell any units today and, if so, how many. In your brief e-mail reply, you will simply answer ``Yes" or ``No" and, if ``Yes" then provide a number of units you want to sell. This number can range from 1 to the total number stored. Be sure to take note of the total units you have stored so that you don't try to sell more than you have. \textit{You must respond no later than 8:00 p.m. Eastern Time every day.}
\end{adjustwidth}
\bigbreak
Before deriving the models, it should first be noted that daily energy prices, which were presented to participants\index{Participant} as ranging from \$0.10 to \$1.50 per unit, were transformed to the values 1 through 15 for model development and data analysis.


\section{Modeling Approaches}\label{1.3}
\subsection{Modeling Prosumer\index{Prosumer} Behavior Using Expected Utility Theory\index{Expected Utility Theory}}

Prosumer\index{Prosumer} participants\index{Participant} in our power-grid scenario\index{Scenario} have to make a new decision every day regarding whether to sell some, all, or none of their surplus energy units to a central utility at the price offered by the utility on that day. To make this decision they must assess the probability that they would receive a higher payoff in the future for those units if they choose not to sell them today.

To model this situation, we assume a range of possible prices that may be offered to the prosumer\index{Prosumer} each day. Further, we assume that the prosumer\index{Prosumer} is provided with a forecast of these prices, i.e., the probabilities of each of these prices occurring on a single day (i.e., the ``single day prices" given in Table \ref{t1}). On each day $d$, a participant\index{Participant} has one or more units of energy to sell. The price on offer that day for a unit of energy is denoted by $i_{d}$. Let $i_{d}^{*}$ denote the unit cutoff price\index{Cutoff price}, which is the lowest price on day $d$ for which the gain from  a sale on that day is greater than the expected gain from a hold on that day. 

The participant\index{Participant} decides whether to sell or hold some or all the units available for sale each day. Expected utilities are computed for sell and hold decisions, the components of which are the gains realized from selling and the potential (future) gains from not selling. Hence, the equation for computing the expected utility of selling a unit of energy on day $d$ at price $i_{d}$ has two components. The first component is the gain realized by selling the unit at price $i_{d}$ on day $d$. This is simply $i_{d}$. The second component is the possible gain that may be realized by holding the unit on day $d$ and selling on a subsequent day. The possible gain from holding is determined by the probability of selling at a higher price on a subsequent day. The probability of selling at a higher price on a subsequent day is determined by the cutoff price\index{Cutoff price} on each subsequent day and the probability of it being met or exceeded by the price on offer on that day. So, the expected utility of selling $n$ units on day $d$ is defined as,

\begin{equation}\label{e1}
\mathbb{E} [U_d(n)]=n{i_{d}}+\left( N_d - n \right)\left(\sum_{j=i^*_{d+1}}^{I}p_{j}j +f_{h, d+1}\sum_{j=1}^{i^*_{d+1}-1}p_{j}\right),
\end{equation}
where $N_d$ is the number of units available for sale on day $d$, $i_d$ is the offered price on day $d$, $i^*_{d+1}$ is the cutoff price\index{Cutoff price} on day $d+1$, $I$ is the highest price ever offered, $p_j$ is the probability that $j$ would be the unit price on offer on any day, and $f_{h,d+1}$ is the expected gain for a hold on day $d+1$ and is defined recursively as:

\begin{equation}\label{e2}
f_{h,d}=\begin{cases}
\sum\limits_{j=i^*_{d+1}}^{I}p_{j}j +f_{h, d+1}\sum\limits_{j=1}^{i^*_{d+1}-1}p_{j}  \ \ \ \ \ d=1,...,D-2\\
\\

\sum\limits_{j=i^*_{D}}^{I}p_{j}j  \ \ \ \ \ \ \ \ \ \ \ \ \ \ \ \ \ \ \ \  \ \ \ \ \ \ \ \ \ \ \ \ \ d=D-1\\
\\

0 \ \ \ \ \ \ \ \ \ \ \ \ \ \ \ \ \ \ \ \ \ \ \ \ \ \ \ \ \ \ \ \ \ \ \ \ \ \ \ \ \ \ d=D,
\end{cases}
\end{equation}
where $D$ is the last day, hence the highest numbered day in the study. If EUT\index{EUT} is considered as a model of human decision making\index{Decision making}, then the best strategy of selling can be stated as,

\begin{equation}\label{e3}
n^*_{d_{EUT}}=\arg\max_{n} \mathbb{E}[U_{d}(n)],
\end{equation}
where $n^*_{d_{EUT}}$ is the value of $n$ that maximizes $\mathbb{E}[U_d(n)]$.

\newtheorem{theorem}{Theorem}
\begin{theorem}EUT\index{EUT} model always predicts either selling all available units or nothing. In other words, $n^*_{d_{EUT}}$ is either equal to $N_d$ or zero.
\end{theorem}

\begin{proof}
If we rearrange $\mathbb{E} [U_{d}(n)]$ as,

\begin{eqnarray}\nonumber \label{e4}
\mathbb{E} [U_{d}(n)] &=& n\left({i_{d}}-\left(\sum_{j=i^*_{d+1}}^{I}p_{j}j +  f_{h,d+1}\sum_{j=1}^{i^*_{d+1}-1}p_{j}\right)\right) \\
&& + N_d\left(\sum_{j=i^*_{d+1}}^{I}p_{j}j +f_{h, d+1}\sum_{j=1}^{i^*_{d+1}-1}p_{j}\right),
\end{eqnarray}
and 

\begin{equation}\label{e5}
s_d \coloneqq {i_{d}}-\left(\sum_{j=i^*_{d+1}}^{I}p_{j}j +f_{h, d+1}\sum_{j=1}^{i^*_{d+1}-1}p_{j}\right),
\end{equation}

\begin{equation}\label{e6}
s_{d_{1}} \coloneqq \sum_{j=i^*_{d+1}}^{I}p_{j}j +f_{h, d+1}\sum_{j=1}^{i^*_{d+1}-1}p_{j}.
\end{equation}
The values of ${i^*_{d}}$ and $p_j$ are always positive. Hence, $s_{d_{1}}$ is always positive for any values of ${i^*_{d}}$ and $p_j$. Therefore,

\begin{equation}\label{e7}
\mathbb{E} [U_{d}(n)]=n\left({i_{d}}-s_{d_{1}}\right)+N_d\left(s_{d_{1}}\right).
\end{equation}
As $N_d$ is always positive, the second term in \eqref{e7} is always positive. If ${i_{d}} < s_{d_{1}}$, the first term in \eqref{e7} becomes negative. Hence, to maximize \eqref{e7} the first term should be minimized. As,
\begin{equation}\label{e8}
n=0, 1,..., N_d,
\end{equation}
$n=0$ minimizes the first term or equivalently maximizes \eqref{e7}. In addition, if ${i_{d}} > s_{d_{1}}$, both terms in \eqref{e7} are positive. Hence, to maximize \eqref{e7}, the first term should be maximized. Considering \eqref{e8}, $n=N_d$ maximizes \eqref{e7}. It follows that,
\begin{equation}\label{e9}
n^*_{d_{EUT}}=\begin{cases}
N_d \ \ \ \ \ \ s_d>0\\
0 \ \ \ \ \ \ \ \  s_d<0. \end{cases}
\end{equation}

\end{proof}
\newtheorem{corollary}{Corollary}
\begin{corollary}
The sell all or nothing result is independent of cutoff prices\index{Cutoff price}, ${i^*_{d}}$.
\end{corollary}

\begin{corollary}
The sell all or nothing result is independent of probabilities of  prices, $p_i$.
\end{corollary}

While the decision resulting from the EUT\index{EUT}-based model is independent of prices, probabilities, and the cutoff price\index{Cutoff price}, numerical evaluation of the expected utility function defined in \eqref{e1} requires all the previously mentioned parameters. While the prices and their probabilities are assumed to be known in the model, we next show how the cutoff prices\index{Cutoff price} involved in the model can be numerically derived.

\subsubsection{Evaluating the Cutoff Price\index{Cutoff price}}

On day $D$, the last day of the study, the probability of selling any units at any price on a subsequent day is zero, hence the expected gain from holding on day $D$ is zero. Therefore, regardless of the price on offer on day $D$, including the lowest possible price, which is 1, a sell decision will be made for all available units on day $D$, i.e.,
\begin{equation}\label{e10}
i^*_{D}=1.
\end{equation}
The expected utility of selling $n$ units at price $i_{D-1}$ can be obtained from \eqref{e1}, \eqref{e2}, and \eqref{e10}. As \eqref{e2} states, when $d = D-1$, $f_{h,D} = 0$. As \eqref{e10} states, $i^*_D = 1$. When these values are inserted in \eqref{e1}, it readily follows that \eqref{e1} reduces to: 

\begin{equation}\label{e11}
\mathbb{E} [U_d(n)]=ni_d+\left( N_d - n \right)\left(\sum_{j=i^*_{d+1}}^{I}p_{j}{j}\right),
\end{equation}
and when $N_{d}=1$, it reduces to,
\begin{numcases}{\mathbb{E} [U_d(n)]=}
i_d  \ \ \ \ \ \ \ \ \ \ \ \ \ \ \ \ \  n=1 \label{e12a}
\\
\sum_{j=i^*_{d+1}}^{I}p_{j}{j}  \ \ \ \ \ \ \ n=0. \label{e12b}
\end{numcases}
We set $i_d$ equal to ascending values beginning with $i^*_D=1$, until we find the lowest value of $i_d$ for which $i_d> \sum_{j=i^*_{d+1}}^{I}p_{j}{j}$. We find $i^*_{D-1}$ by,

\begin{equation}\label{e13}
i^*_{D-1} = \lceil i_d \rceil, 
\end{equation}
where $\lceil . \rceil$ denotes ceiling function. For $d<D-1$, the expected gain from holding for a single unit can be given by the recursive Equation \eqref{e2}. When $N_{d}=1$, the expected utility reduces to,

\begin{numcases}{\mathbb{E} [U_d(n)]=}
i_d  \ \ \ \ \ \ \ \ \ \ \ \ \ \ \ \ \ \ \ \ \ \ \ \ \ \ \ \ \ \ \ \ \ \ \ \ \ \ \ \ \ \  \ \ \ \ n=1 \label{e14a}
\\
\sum_{j=i^*_{d+1}}^{I}p_{j}j +\left(\sum_{j=1}^{i^*_{d+1}-1}p_{j}\right)    f_{h, d+1} \ \ \ \ \ \ \ \ \ n=0,\label{e14b}
\end{numcases} 
where \eqref{e14a} is the gain for sell when $n=1$ and \eqref{e14b} is the expected gain for hold when $n=0$. To compute $i^*_d$, we find the smallest value of $i_d$ such that,

\begin{equation}\label{e15}
i_d>\sum_{j=i^*_{d+1}}^{I}p_{j}j +\left(\sum_{j=1}^{i^*_{d+1}-1}p_{j}\right) f_{h, d+1},
\end{equation}
i.e., $i^*_{d}$ is

\begin{equation}\label{e16}
i^*_{d} = \lceil i_d \rceil,
\end{equation}
for which \eqref{e15} is satisfied. To use  \eqref{e14b} to compute the value of $i^*_d$ we need the values of $i^*_{d+1}$ through $i^*_D$. Hence, to compute the value of each $i^*_d$ we begin with $i^*_D$ then compute the value of the cutoff price\index{Cutoff price} for each successively earlier day in the study.

\subsection{Bounded Horizon Model of Prosumer\index{Prosumer} Behavior}
In the bounded horizon model (denoted as EUT\textsubscript{TW}), a prosumer\index{Prosumer} may compute the probability of selling at a higher price over a fixed number of days called a time window\index{Time window} until near the end of the study when the remaining days are fewer than the number of days in the prosumer's time window\index{Time window}. Let $ \tilde{d}=D-d$, be the number of remaining days. So ${i^*_{w}}$ is the cutoff price\index{Cutoff price} that is the lowest price for which the gain from a sale is greater than the expected gain from a hold for a time window\index{Time window} of $w$ days, where $w=\min (\tilde{d}, t)$, and $t$, the time window\index{Time window}, is the number of days over which expected gain for a hold is computed by a participant\index{Participant}. The expected utility of selling $n$ units on day $d$ with a time window\index{Time window} of $w$ is defined as,
\begin{equation}\label{e17}
\mathbb{E} [U_{d,w}(n)]=n{i_{d}}+\left( N_d - n \right)\left(\sum_{j=i^*_{w-1}}^{I}p_{j}j +f_{h, w-1}\sum_{j=1}^{i^*_{w-1}-1}p_{j}\right),
\end{equation}
where $f_{h,w}$ is the expected gain for hold on day $d$ with a time window\index{Time window} of $w$ and is defined as,

\begin{equation}\label{e18}
f_{h,w}=\begin{cases}
\sum\limits_{j=i^*_{w-1}}^{I}p_{j}j +f_{h, w-1}\sum\limits_{j=1}^{i^*_{w-1}-1}p_{j}  \ \ \ \ \ \ \ \ \ \ \  w>1\\
\\
\sum\limits_{j=i^*_{w-1}}^{I}p_{j}j  \ \ \ \ \ \ \ \
\ \ \ \ \ \ \ \ \ \ \ \ \ \ \ \ \ \ \ \ \ \ \ \ \ \ w=1\\
\\
0 \ \ \ \ \ \ \ \ \ \ \ \ \ \ \ \ \ \ \ \ \ \ \ \ \ \ \ \ \ \ \ \ \ \ \ \ \ \ \ \ \ \ \ \ \ w=0.
\end{cases}
\end{equation}
If EUT\textsubscript{TW} is considered as a model of human decision making\index{Decision making}, then the best strategy of selling can be stated as,

\begin{equation}\label{e19}
n^*_{{d,w}_{EUT}}=\arg\max_{n} \mathbb{E}[U_{d,w}(n)].
\end{equation}

\begin{theorem} EUT model with a time window\index{Time window} always predicts either selling all available units or nothing. In other words, $n^*_{{d,w}_{EUT}}$ is either equal to $N_d$ or zero.
\end{theorem}
\begin{proof}
If we rearrange $\mathbb{E}[U_{d,w}(n)]$ as,

\begin{eqnarray}\nonumber \label{e20}
\mathbb{E} [U_{d,w}(n)]&=&n\left({i_{d}}-\left(\sum_{j=i^*_{w-1}}^{I}p_{j}j +f_{h, w-1}\sum_{j=1}^{i^*_{w-1}-1}p_{j}\right)\right)+\\
&& N_d\left(\sum_{j=i^*_{w-1}}^{I}p_{j}j +f_{h, w-1}\sum_{j=1}^{i^*_{w-1}-1}p_{j}\right),
\end{eqnarray}

\begin{equation}\label{e21}
s_w \coloneqq {i_{d}}-\left(\sum_{j=i^*_{w-1}}^{I}p_{j}j +f_{h, w-1}\sum_{j=1}^{i^*_{w-1}-1}p_{j}\right),
\end{equation}
and
\begin{equation}\label{e22}
s_{w_{1}} \coloneqq \sum_{j=i^*_{w-1}}^{I}p_{j}j +f_{h, w-1}\sum_{j=1}^{i^*_{w-1}-1}p_{j}.
\end{equation}
As ${i^*_{w}}$ and $p_j$ are always positive, for any values of ${i^*_{w}}$ and $p_j$, $s_{w_{1}}$ is always positive. Therefore,

\begin{equation}\label{e23}
\mathbb{E} [U_{d,w}(n)]=n\left({i_{d}}-s_{w_{1}}\right)+N_d\left(s_{w_{1}}\right).
\end{equation}
As $N_d$ is always positive, the second term is always positive. If ${i_{d}} < s_{w_{1}}$, the first term in \eqref{e23} is negative. Therefore, to maximize \eqref{e23}, the first term should be minimized. Considering \eqref{e8}, $n=0$ makes the first term minimized or equivalently makes \eqref{e23} maximized. 
In addition, If ${i_{d}} > s_{w_{1}}$, both terms in \eqref{e23} are positive. Hence, to maximize \eqref{e23}, the first term should be maximized. It is clear that $n=N_d$ makes the first term maximized or equivalently makes \eqref{e23} maximized. It follows that,
\begin{equation}\label{e24}
n^*_{{d,w}_{EUT}}=\begin{cases}
N_d \ \ \ \ \ \ s_w>0\\
0 \ \ \ \ \ \ \ \  s_w<0. \end{cases}
\end{equation}

\end{proof}

\begin{corollary}
The sell all or nothing result is independent of cutoff prices\index{Cutoff price}, ${i^*_{w}}$.
\end{corollary}

\begin{corollary}
The sell all or nothing result is independent of probabilities of  prices, $p_i$.
\end{corollary}

\subsubsection{Evaluating the Cutoff Prices\index{Cutoff price}}

Following the approach outlined in the previous section, the cutoff prices\index{Cutoff price} in the EUT\textsubscript{TW} model can be obtained as follows. When Equations \eqref{e14a} and \eqref{e14b} are used to compute $i^*_d$, this is also the cutoff price\index{Cutoff price} for $t=D-d$. The $D-d$ days remaining in the study are a time window\index{Time window} of $D-d$ days. Hence, when $\tilde{d}>t$, $t$ is used. However, when $\tilde{d}<t$ it makes no sense to use $t$ because the probability of a sale beyond $\tilde{d}$ days is zero. When $\tilde{d} \leq t$ the EUT\textsubscript{TW} model becomes the EUT\index{EUT} model. So $\tilde{d}$ is used to compute the cutoff price\index{Cutoff price}. So $w=\min (\tilde{d}, t)$ and is used instead of $t$. Therefore, when $N_{d}=1$ the expected utility reduces to:

\begin{numcases}{\mathbb{E} [U_{d,w}(n)]=}
i_d  \ \ \ \ \ \ \ \ \ \ \ \ \ \ \ \ \ \ \ \ \ \ \ \ \ \ \ \ \ \ \ \ \ \ \ \ \ \ \ \ \ \ \ \ \ \ \ n=1 \label{e25a}
\\
\sum_{j=i^*_{w-1}}^{I}p_{j}j +\left(\sum_{j=1}^{i^*_{w-1}-1}p_{j}\right)  f_{h, w-1} \ \ \ \ \ \ \ \ n=0,\label{e25b}
\end{numcases} 
where \eqref{e25a} is the gain from a sale ($n=1$) and \eqref{e25b} is the expected gain from a hold ($n=0$). Notice that \eqref{e25a} and \eqref{e25b} are identical to \eqref{e14a} and \eqref{e14b} except for the substitution of $w$, which is decremented, for $d$, which is incremented. Hence, when $w=\tilde{d}$ then \eqref{e25a} and \eqref{e25b} are identical to \eqref{e14a} and \eqref{e14b} except for notation. $\tilde{d}=D-d$ and so $i^*_d$ and $i^*_{w=\tilde{d}}$ are different ways of denoting the cutoff value for the same day. However, when $w=t$ then \eqref{e25a} and \eqref{e25b} computes the cutoff price\index{Cutoff price} for a smaller number of days than is used by \eqref{e14a} and \eqref{e14b}. When $w = 0$, whatever units are available for sale will be sold at whatever price is offered.

\begin{theorem} Frequency of selling in EUT with windowing is greater than EUT\index{EUT}.
\end{theorem}

\begin{proof}
$f_{h,w}$ in \eqref{e18} is an increasing function of $w$; that is, the gain from holding increases as window size increases. Reconsidering \eqref{e23}, the second term is always positive, and because $s_{w_{1}}$ is an increasing function of $w$, therefore, the first term is a decreasing function of $w$. Consequently, for a constant $N_d$ and a specific ${i_{d}}$, when time window\index{Time window} is being increased, the first term becomes smaller. Thus, it reduces the frequency of selling all or it increases the frequency of selling nothing. 
\end{proof}

\section{Data Fitting\index{Data fitting}, Results, and Analysis}\label{1.4}

Table \ref{t2} shows the cutoff price\index{Cutoff price} as a function of the number of days remaining in the study, computed from Equations \eqref{e15} and \eqref{e16}. As shown in the table, the cutoff price\index{Cutoff price} necessarily declines as the study approaches its end and the number of days on which a higher price might be offered is reduced. Fig. \ref{f1} shows the cutoff price\index{Cutoff price} as a function of both time window\index{Time window} and number of days remaining in the simulation.

Note that while the study took place over a 10-week period, no data was collected on two of those days (one holiday and one day on which the data was corrupted). Thus, the following analyses are based on 68 days of data collection.

\begin{table}[b]
\centering
\centering
\caption{Number of remaining days in the study for which each cutoff price applies.}
\begin{tabular}{c|c}
Cutoff price, $i^*_{d}$ & Days remaining, $\tilde{d}$ \\
\hline
 14& $\geq$31 \\
13 & 16--30    \\
12 & 9--15     \\
11 & 6--8      \\
10 & 4--5      \\
9  & 3        \\
8 & 2         \\
7 & 1         \\
1 & 0          
\end{tabular}
\label{t2}
\end{table}

\begin{figure}[t]
\centering
\includegraphics[width=4in]{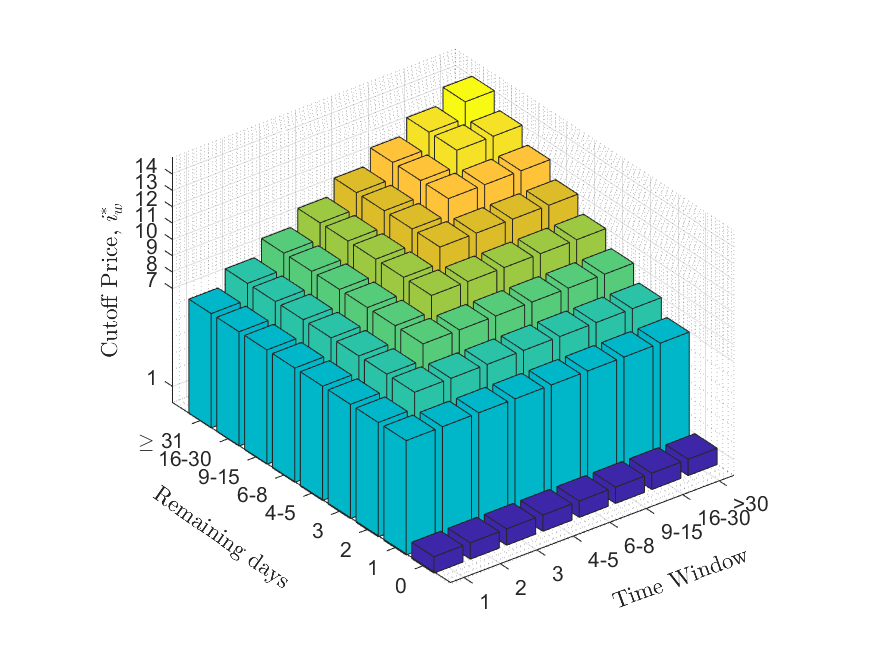}
\caption{Initial cutoff price as a function of time window (TW) and number of days remaining in simulation.}
\label{f1}
\end{figure}

\subsection{Data Fitting\index{Data fitting} with Mean Deviation}
For each participant\index{Participant}, for each day of the study, let $N_d$ be the number of units available for sale, $n^*_d$ be the number of units that a model predicted should be sold, and $n_d$ be the number of units sold. Obviously, the deviation between the proportions of predicted and observed units sold is only meaningful when $N_d> 0$. When $N_d> 0$, the deviation between the proportions of predicted and observed units sold is $|\frac{n^*_d}{N_d}-\frac{n_d}{N_d}|$. The mean deviation for all the days of the study on which $N_d > 0$ is the fit of the model. So, the time window\index{Time window} that generates the smallest mean deviation for a participant\index{Participant} is the best-fitting time window\index{Time window} for that participant\index{Participant}. The mean deviation for EUT\index{EUT} can be stated as,
\begin{equation}\label{e26}
MD=\frac{\sum_{d \in D'}|\frac{n^*_d}{N_d}-\frac{n_d}{N_d}|}{|D'|},
\end{equation}
and for EUT\textsubscript{TW} is,
\begin{equation}\label{e27}
MD=\frac{\sum_{d \in D'}|\frac{n^*_{d,w}(w)}{N_d}-\frac{n_d}{N_d}|}{|D'|},
\end{equation}
where the set $D'=\{\,d \mid N_d>0\,\}$, and $|D'|$ is the number of days on which $N_d > 0$. Consequently, the optimal time window\index{Time window} is fit to the data using
\begin{equation}\label{e28}
w^*=\arg\min_{w} MD.
\end{equation}

\begin{figure}[t]
\centering\includegraphics[width=4in]{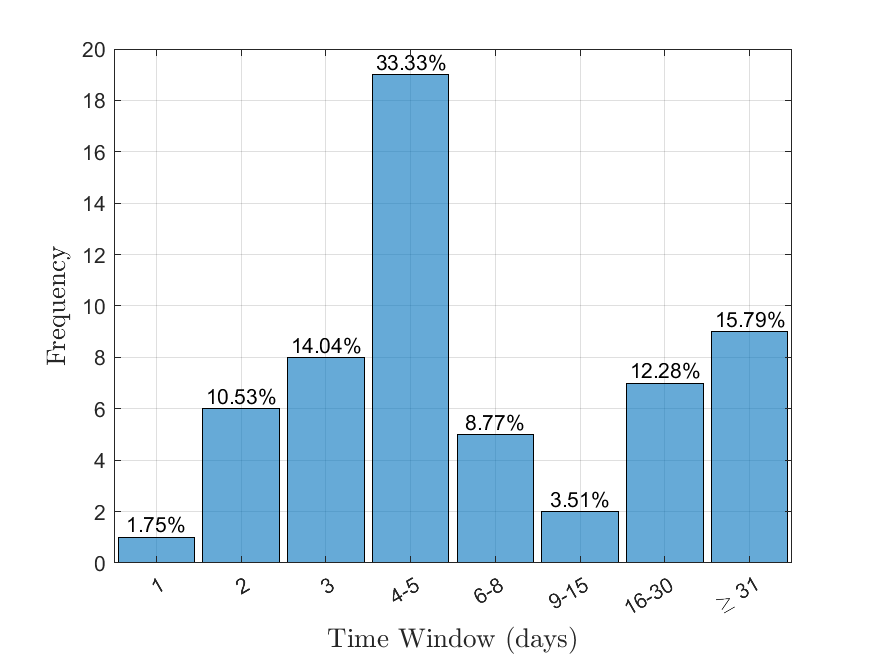}
\caption{The percent of participants whose responses were best fit by each time window.}
\label{f2}
\end{figure}

The responses of each participant\index{Participant} were compared with those predicted by the cutoff prices\index{Cutoff price} for a sell decision for each time window\index{Time window}, as shown in Fig. \ref{f1}. Fig. \ref{f2} shows the histogram frequency distribution over participants\index{Participant} for the size of the best-fitting time window\index{Time window}.

EUT\index{EUT} does not include the concept of a time window\index{Time window}. For EUT\index{EUT} on each day of the study the probability that each price would be the highest would be computed for the remaining days of the study. For example, because of the distribution of daily price probabilities in this study, until there were only 30 days remaining in the study (a range of from 67 to 31 days remaining) a cutoff price\index{Cutoff price} of 14 was predicted.

EUT\index{EUT} predicted a sell decision on only four days of the study, including a day when the price offered was a maximum of 15 and the last day of the study, when selling is trivially predicted by every model. Hence, EUT\index{EUT} made only two nontrivial sell predictions. In fact, 56 of the 57 participants\index{Participant} sold more than four days as shown in Fig. \ref{f13}. The mean number of sell decisions for all participants\index{Participant} was 20 with a standard deviation of 11.5, and a range from 3 to 56.

\begin{figure}[t]
\centering\includegraphics[width=4in]{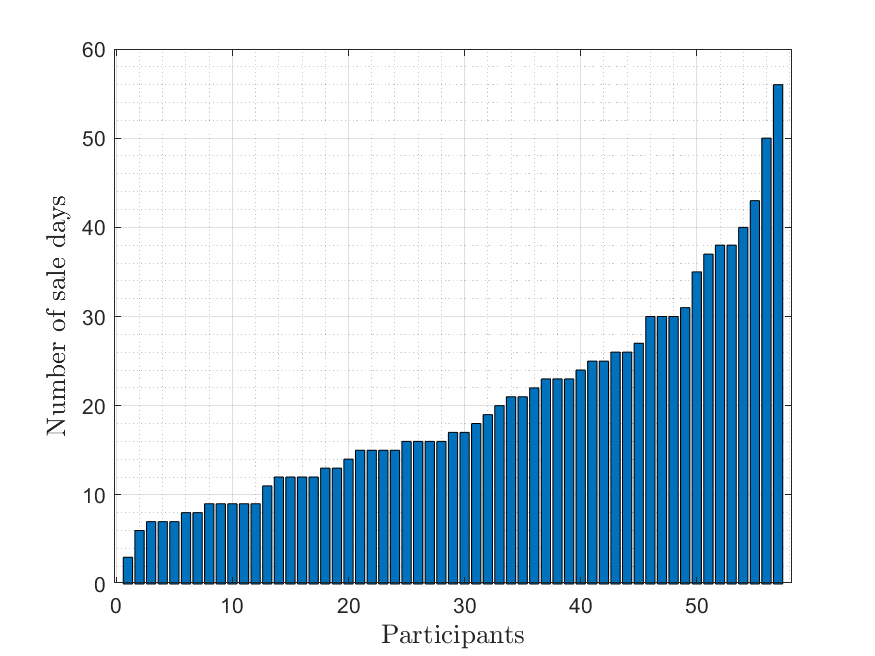}
\caption{The number of sell decisions for each participant in 10-week study.}
\label{f13}
\end{figure}

The frequency with which our participants\index{Participant} made sell decisions demonstrated that most of them adopted time windows\index{Time window} over which the probability of a gain was computed rather than using the remaining days in the study to compute the gain from a hold. As can be seen from Fig. \ref{f2}, the distribution of time windows\index{Time window} was bimodal with peaks at 4--5 days and at 31--68 days. However, the peaks are not of equal size. Only 16\% of the participants\index{Participant} were best fit by a time window\index{Time window} ranging from 31--68 days, which is the range consistent with EUT\index{EUT}. Data from 84\% of the participants\index{Participant} were best fit using time windows\index{Time window} shorter than 31 days; the great majority of those appearing to use time windows\index{Time window} of less than 9 days in determining whether to sell or hold.

Fig. \ref{f3} shows, for each model, the mean deviation between the number of units predicted to be sold and the number of units actually sold by a participant\index{Participant}. Each of the 57 participants\index{Participant} is represented in the figure by a pair of points--one for the average deviation of their data from the EUT\index{EUT} prediction and the other for average deviation from the EUT\textsubscript{TW} prediction with the best fit of model parameter $w$. Participants\index{Participant} are ordered on the $x$ axis according to the number of days on which they sold at least one unit, from the lowest number of sale days to highest. As can be seen clearly from the plot, EUT\textsubscript{TW} provides a better fit than EUT\index{EUT}.

\begin{figure}[t]
\centering\includegraphics[width=4in]{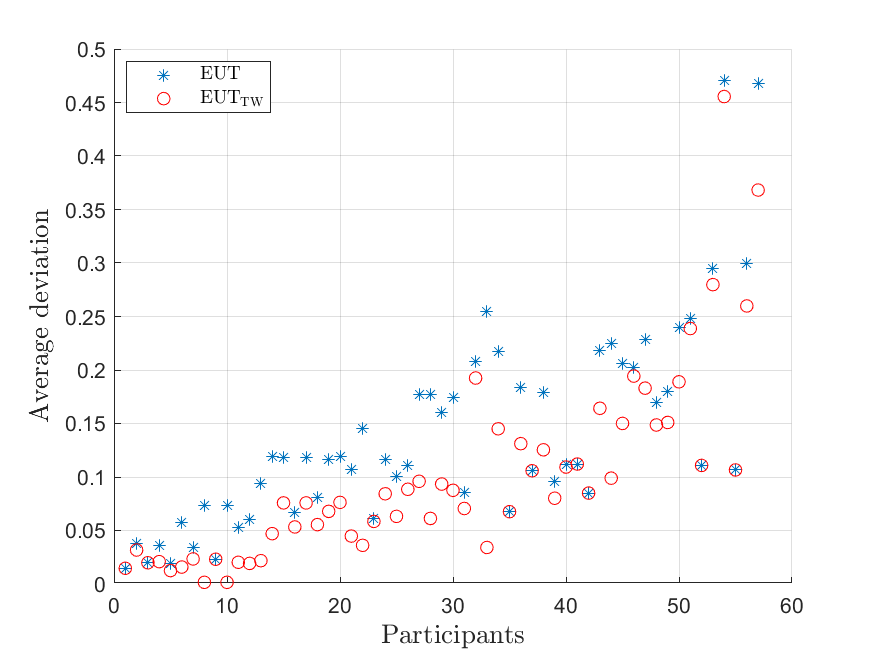}
\caption{Average deviation for EUT and average deviation for EUT with time window}
\label{f3}
\end{figure}

\subsection{Data Fitting\index{Data fitting} with Proportional Deviation}

A more detailed analysis of the result can be made by computing the proportional deviation between the number of units predicted to be sold by a model and the number of units actually sold by a participant\index{Participant}. For a participant\index{Participant}, for each day of the study, the absolute difference between the number of units predicted to be sold and the number of units actually sold is computed. Then, the sum of the deviations is divided by the sum of the maximum of either the number of units predicted to be sold or the number of units sold. Hence, the proportional deviation for EUT\index{EUT} can be stated as, 
\begin{equation}\label{e29}
PD=\frac{\sum_{d \in D'}|n^*_d-n_d|}{\sum_{d \in D'}\max(n^*_d,n_d)}, 
\end{equation}
and for EUT\textsubscript{TW} is,
\begin{equation}\label{e30}
PD=\frac{\sum_{d \in D'}|n^*_{d,w}(w)-n_d|}{\sum_{d \in D'}\max(n^*_{d,w}(w),n_d)}.
\end{equation}
Therefore, the optimal time window\index{Time window} is fit to the data using
\begin{equation}\label{e31}
w^*=\arg\min_{w} PD.
\end{equation}

\begin{figure}[t]
\centering\includegraphics[width=4in]{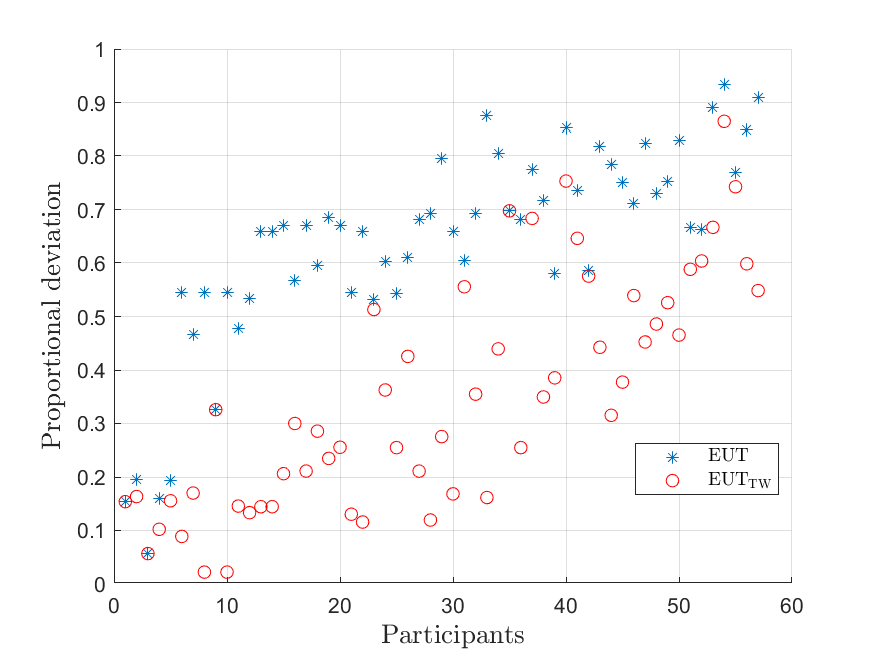}
\caption{Proportional deviation for EUT and proportional deviation for EUT with time window}
\label{f4}
\end{figure}
 \begin{figure}[t]
\centering\includegraphics[width=4in]{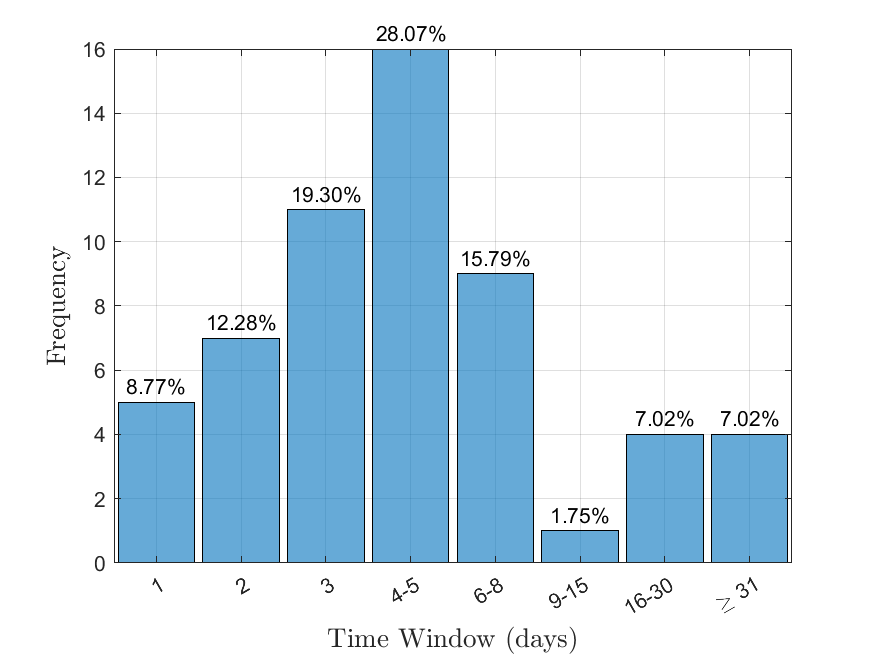}
\caption{The percent of participants whose responses were best fit by each time window based on the proportional deviation method.}
\label{f5}
\end{figure}
The proportional deviation is a number between 0 and 1, where 0 is a perfect fit of the model predictions to the results. Fig. \ref{f4} shows the proportional deviation for EUT\index{EUT} and proportional deviation for EUT\textsubscript{TW} (with the best fit of model parameter $w$) for each participant\index{Participant} in the study, ordered from those participants\index{Participant} who made the fewest number of sales in the grid game\index{Grid game} to those who made the most. These results demonstrate even more clearly than those plotted in Fig. \ref{f3} that the concepts of horizon and time window\index{Time window} could be essential to the understanding of ongoing markets, including the energy market\index{Energy market} studied here.

Fig. \ref{f5} shows the histogram frequency distribution over participants\index{Participant} for the size of the best-fitting time window\index{Time window}. This plot is obtained based on fitting procedure for a time window\index{Time window} through proportional deviation.

When a large payout has a low daily probability but necessarily has a high probability of occurring over the long term, human participants\index{Participant} do not wait around for it to occur; almost never participating in the market at all, as EUT\index{EUT} prescribes. Rather, for sell decisions, they make use of time windows\index{Time window} that are undoubtedly related in some meaningful way to their lives. As can be seen, EUT\textsubscript{TW} necessarily provides a better fit than EUT\index{EUT}. Because EUT\index{EUT} is a special case of EUT\textsubscript{TW}, at worst, the fits of the two models would be identical. However, whenever a participant\index{Participant} in the grid game\index{Grid game} sold more than the three times predicted by EUT\index{EUT}, the fit for EUT\textsubscript{TW} was better than for EUT\index{EUT}.

\section{Discussion}\label{1.5}

Given the parameters defining this study's smart grid\index{Smart grid} scenario\index{Scenario}, EUT\index{EUT} predicts that rational participants\index{Participant} would sell their surplus energy units on only four days in the study (the three days when the price exceeded the cutoff price\index{Cutoff price} and on the final day).  Furthermore, EUT\index{EUT} predicts that our prosumers\index{Prosumer} should always sell either none of their units or all of them.  Even a cursory analysis of the data revealed that our participants' sell/hold decisions were not aligned with EUT\index{EUT}. Many participants\index{Participant} sold much more frequently than EUT\index{EUT} predicted and often sold only a portion of their energy units.

The addition of a time window\index{Time window} parameter to EUT\index{EUT} allows for a better fit to our participants' selling behavior.
The bimodal nature of the distribution of best-fitting time windows\index{Time window} in Fig. \ref{f2} and, to a lesser extent, in Fig. \ref{f5} obviates the individual differences among our participants\index{Participant}.  A few of those near the far right of the distributions (with time windows\index{Time window} exceeding 31 days), are fit just as well by traditional EUT\index{EUT} as by EUT\textsubscript{TW}. Based upon the demographic information gathered prior to the study and on a poststudy questionnaire, these few individuals had slightly more years of education and tended to have taken more math and science courses in school, although neither of these trends were statistically significant.  The relationship of decision behavior to numeracy and other cognitive measures would be interesting to explore further in future research.

Parallels between our participants' decision behavior and that of investors in some situations are worth exploring further. The underlying mechanism for the ``sell too soon" side of the disposition effect may be similar to that which led many of our smart grid\index{Smart grid} participants\index{Participant} to sell energy units much more frequently (and less optimally) than would be predicted by EUT\index{EUT}. The finding that most of our study participants\index{Participant} often made their sell/hold decisions with a relatively short, self-imposed time window\index{Time window} is also conceptually similar to the previously-mentioned MLA phenomenon. To the extent that this myopia is temporal in nature--i.e., due to adopting a short-term horizon when longer-term would be more appropriate--it seems similar to the bounded time window\index{Time window} many of our grid study participants\index{Participant} adopted for their sell/hold decision.

Another related concept from investing and behavioral economics is ``narrow framing" \cite{Tha1,Tha2}. Narrow framing holds that investors (or others making decisions under uncertainty) evaluate a given sell/hold decision or other gamble in isolation rather than evaluating it in the context of their entire portfolio and/or a longer-term horizon.  That is, they fail to appreciate ``the big picture."  It may be useful to think of MLA and the bounded time window\index{Time window} findings of the present study as instances of narrow framing.

Two other research paradigms where temporal horizons\index{Temporal horizon} are an important factor in decision making\index{Decision making} are the previously-mentioned job search studies and the temporal discounting literature. In their experiments loosely simulating a job search scenario\index{Scenario}, Cox and Oaxaca \cite{Cox1,Cox2} found that study participants\index{Participant} terminated their searches sooner than predicted by a risk-neutral linear search model.  The data were fit significantly better by a risk-averse model.  Their subjects' tendency to decide to terminate their searches and accept suboptimal job offers somewhat parallels our participants' tendency to choose suboptimal, shorter time windows\index{Time window} for selling their energy units.  However, there are too many differences in the experimental procedures and analytical approaches across the studies to conclude that the underlying behavioral processes are the same.

The temporal (or delay) discounting literature offers another research framework in which to consider findings of the present study. In the typical temporal discounting example, an individual given a choice between a small reward delivered now or soon, and a relatively larger award delivered at a later time, will tend to choose the smaller, sooner reward.  Of course, many factors can influence the shape of the discounting function, and can even reverse the preference for smaller--sooner over larger--later. There are excellent reviews of this literature from an economic perspective \cite{Frederick2002} and from a more psychological perspective \cite{Green2004ADF}.  It is tempting to postulate that the choice of relatively short time windows\index{Time window} and relatively smaller rewards (unit prices) by many of our grid study participants shares the same underlying decision-making mechanism with that which leads to temporal discounting. Once again though, large differences in experimental paradigms and modeling approaches make it difficult to directly compare results. Green and Myerson \cite{Green2004ADF} cite extensive evidence for the similarities between temporal discounting and the discounting of probabilistic rewards.  Although they tentatively conclude that different underlying mechanisms are at work, their attempt to study both phenomena within a common research framework is laudable.


%





\ifCLASSOPTIONcaptionsoff
  \newpage
\fi



\bibliographystyle{IEEEtran}
\bibliography{IEEEfull,bare_jrnl}
\end{document}